\documentclass[12pt,draftcls,onecolumn]{IEEEtran}
\ifCLASSINFOpdf
\else
\fi
\usepackage{epsfig}
\usepackage{amssymb}
\usepackage{amsmath}
\usepackage{amsfonts}
\usepackage{color}
\newtheorem{theorem}{Theorem}

\newtheorem{definition}{Definition}
\newtheorem{lemma}{Lemma}

\newtheorem{remark}{Remark}
\newtheorem{example}{Example}
\newtheorem{assumption}{Assumption}
\newtheorem{proof}{Proof}
\newtheorem{proof of Theorem 1}{Proof of Theorem 1}
\newcommand*{\QEDA}{\hfill\ensuremath{\blacksquare}}

\begin{document}
%
% paper title
% can use linebreaks \\ within to get better formatting as desired

\title{Distributed algorithms for solving convex inequalities}
%
%
% author names and IEEE memberships
% note positions of commas and nonbreaking spaces ( ~ ) LaTeX will not break
% a structure at a ~ so this keeps an author's name from being broken across
% two lines.
% use \thanks{} to gain access to the first footnote area
% a separate \thanks must be used for each paragraph as LaTeX2e's \thanks
% was not built to handle multiple paragraphs
%

\author{Kaihong~Lu,~%\IEEEmembership{Member,~IEEE,}
       Gangshan~Jing,~%\IEEEmembership{Member,~IEEE}
       and~Long~Wang~%\IEEEmembership{Member,~IEEE}
        %and~Jane~Doe,~\IEEEmembership{Life~Fellow,~IEEE}% <-this % stops a space
\thanks{This work was supported by National Science Foundation of China (Grant Nos. 61533001 61375120, and 61603288). (\emph{Corresponding author: Long Wang})}
\thanks{K. Lu and G. Jing are with the Center for Complex Systems, School of Mechano-electronic Engineering, Xidian University, Xi'an 710071, China
 (e-mail:khong\_lu@163.com; nameisjing@gmail.com)
 %e-mail: (see http://www.michaelshell.org/contact.html).
 }% <-this % stops a space
\thanks{L. Wang is with the Center for Systems and Control, College of Engineering, Peking
University, Beijing 100871, China
 (e-mail: longwang@pku.edu.cn)}
% <-this % stops a space
%\thanks{Manuscript received April 19, 2005; revised January 11, 2007.}
}

\maketitle

\begin{abstract}
%\boldmath
In this paper, a distributed subgradient-based algorithm is proposed for continuous-time multi-agent systems to search a feasible solution to convex inequalities. The algorithm involves each agent achieving a state constrained by its own inequalities while exchanging local information with other agents under a time-varying directed communication graph. With the validity of a mild connectivity condition associated with the communication graph, it is shown that all agents will reach agreement asymptotically and the consensus state is in the solution set of the inequalities. Furthermore, the method is also extended to solving the distributed optimization problem of minimizing the sum of local objective functions subject to convex inequalities. A simulation example is presented to demonstrate the effectiveness of the theoretical results.
\end{abstract}

\begin{keywords}
Multi-agent systems; Convex inequalities; Consensus; Distributed optimization.
\end{keywords}

% For peer review papers, you can put extra information on the cover
% page as needed:
% \ifCLASSOPTIONpeerreview
% \begin{center} \bfseries EDICS Category: 3-BBND \end{center}
% \fi
%
% For peerreview papers, this IEEEtran command inserts a page break and
% creates the second title. It will be ignored for other modes.
\IEEEpeerreviewmaketitle

\section{Introduction}

\IEEEPARstart{D}{istributed} coordination problems of multi-agent systems (MASs) have been intensively investigated in various areas including engineering, natural science, and social science \cite{RenA25}-\cite{MaJ28}. As a fundamental coordination problem, the consensus problem which requires that a group of autonomous agents achieve a common state has attracted much attention, see \cite{Jing26}-\cite{Y. Zheng04}. This is due to its wide applications in distributed control and estimation \cite{S. Kar05}, distributed signal processing \cite{Dimakis06}, and distributed optimization \cite{A. Nedi07}-\cite{Gharesifard12}.

Consensus-based algorithms have been effectively used for solving linear algebraic equations \cite{S. Mou13}-\cite{G. Shi18}. The natural idea for solving large-scale linear algebraic equations is to decompose them into smaller ones which can then be solved by a multi-agent network \cite{S. Mou13}. By using orthogonal projection theory, the problem of solving linear equations has been converted to a consensus problem of MASs in the literature. In \cite{S. Mou15}, linear equations with a unique solution were solved by multiple agents under a fixed undirected communication graph. In \cite{S. Mou16}, linear equations with multiple solutions were further investigated under a time-varying directed communication graph. However, there is a limitation that the algorithms in \cite{S. Mou15, S. Mou16} require the initial value of each agent's state to satisfy its equation constraints. In order to overcome this problem, the ``consensus + projection" and distributed projected consensus algorithms were proposed to solve linear equations by \cite{G. Shi18}, where they project each agent's state into the affine subspace specified by its own equation constraints, then solving the equation is equivalent to finding a point in the intersection of all the affine subspaces.

Similar to solving linear equations, searching feasible solutions to a set of algebraic inequalities is also a significant problem that remains to be dealt with. Some simple inequalities could be solved for trivial solutions by transforming them to equations. However, for complex and large-scale inequalities, transforming them to special equations requires a vast amount of computations and may cause the equations having no solution. In fact, multi-agent systems are often subjected to state constraints. For instance, in formation control, containment control and alignment problems, each agent's position is usually limited to stay in a certain region. In this note, we consider the constraints as convex inequalities. Inspired by the distributed methods for solving linear equations \cite{S. Mou13, S. Mou16, G. Shi18}, we solve the convex inequalites in a distributed manner. Different from the investigations \cite{S. Mou15}-\cite{H. Cao17} associated with solving linear equations, in problems of solving convex inequalities, the restriction of agents' initial states leads to the reduction of the feasible region. Moreover, the solution space of convex inequalities is not an affine subspace, which implies that the method in \cite{G. Shi18} is not applicable.

{Recently, some significant results on distributed algorithms combining consensus and subgradient algorithms were published. In \cite{A. Nedi07, A. Nedi11}, the ``consensus + subgradient" algorithm was used to minimize a sum of convex functions via an agent network. In \cite{Nedi08}, a distributed projected subgradient algorithm was proposed to solve the constrained optimization problem, where each agent should keep a state lying in its own convex set.
This algorithm with time delays was studied in \cite{Lin10}.}

{Investigations in \cite{A. Nedi07}-\cite{Lin10} are all conducted for discrete-time MASs. Nevertheless, agents are often modeled by continuous-time dynamic systems in practical applications of motion coordination control. For example, in rendezvous problems, multiple vehicles that are required to reach a desired common location usually have continuous-time dynamics \cite{McLain13}. Moreover, the results on discrete-time distributed algorithms can not be directly applied to the continuous-time cases. In fact, some distributed gradient algorithms have been proposed for continuous-time MASs under fixed graphs \cite{Tang11}-\cite{Gharesifard12}. Different from them, we investigate the distributed subgradient-based algorithm for continuous-time MASs in the scenarios that the graph is time-varying.}

In this note, we present a distributed subgradient-based algorithm to search a feasible solution to convex inequalities via a network of continuous-time agents, which enables all agents' states to approximate to a common point in the solution set of inequalities. By implementing the algorithm, each agent adjusts its state value based on its own inequality information and the local information received from its immediate neighbors. The underlying communication graph is modeled as a time-varying directed graph. We show that if the $\delta-$graph, induced by the time-varying directed graph, is strongly connected, all agents' states will reach a common point asymptotically and the point is a feasible solution to convex inequalities. Moreover, this method will be extended to solving the distributed optimization problem of minimizing the sum of local objective functions subject to convex inequalities. Numerical simulations are provided to demonstrate the effectiveness of our theoretical results.

This note is organized as follows. In Section \ref{se1}, we formulate the problem to be studied and present the distributed algorithm for continuous-time multi-agent systems. In Section \ref{se2}, we state our main result and give its proof in detail. In Section \ref{se3}, we extend our method to solving the distributed optimization problem of minimizing the sum of local objective functions subject to convex inequalities. In Section \ref{se4}, Simulation examples are presented. Section \ref{se5} concludes the whole paper.

{\bf Notation:} Throughout this note, we use $|x|$ to represent the absolute value of scalar $x$. The operator $\lfloor x\rfloor$ is used to denote the largest integer not larger than the value of $x$. $\mathbb{R}$ and $\mathbb{N}$ denote the set of real number and the set of positive  integer,respectively. Let $\mathbb{R}^{m}$ be the $m$-dimensional real vector space. For a given vector $\textbf{x}\in\mathbb{R}^{m}$, $\textbf{x}\leq0$ implies that each entry of vector $\textbf{x}$ is not greater that zero. $\|\textbf{x}\|$ denotes the standard Euclidean norm, i.e., $\|\textbf{x}\|=\sqrt{\textbf{x}^T\textbf{x}}$. And $\|\textbf{x}\|_1$  is used to denote the 1-norm, i.e., $\|\textbf{x}\|_1=\sum\limits_{j=1}^{m}|\textbf{x}_i|$, where $\textbf{x}_i$ represents the $i^{th}$ entry of vector $\textbf{x}$. For any two vectors $\textbf{u}$ and $\textbf{v}$, the operator $\langle\textbf{u}, \textbf{v}\rangle$ denotes the inner product of $\textbf{u}$ and $\textbf{v}$. $\textbf{1}\in \mathbb{R}^m$ denotes the $m$-dimensional vector with elements being all ones. For a matrix $\textbf{A}$, $[\textbf{A}]_{ij}$ denotes the matrix entry in the $i^{th}$ row and $j^{th}$ column, $[\textbf{A}]_{i\cdot}$ represents the $i^{th}$ row of the matrix $\textbf{A}$, and $[\textbf{A}]_{\cdot j}$ represents the $j^{th}$ column of the matrix $\textbf{A}$. For set $\mathcal{K}\subset\mathbb{R}^n$, we use $P_\mathcal{K}[\cdot]$ to denote a projection operator given by $P_\mathcal{K}(\textbf{u})=\arg\min_{\textbf{v}\in \mathcal{K}}\|\textbf{u}-\textbf{v}\|$.

%%%%%%%%%%%%%%%%%%%%%%%%%%%%%%%%%%%%%%%%%%%%%%%%%%%%%%%%%%%%%%%%%%%%%%%%%%%%%%
\section{Problem formulation}\label{se1}

\subsection{Basic graph theory}

The time-varying directed communication topology is denoted by $\mathcal{G}(t)=(\mathcal{V},\mathcal{E}(t), \mathcal{A}(t))$. $\mathcal{V}$ is a set of vertex, $\mathcal{E}(t)\subset\mathcal{V}\times\mathcal{V}$ is an edge set, and the weighted matrix $\mathcal{A}(t)=(a_{ij}(t))_{n\times n}$ is a non-negative matrix for adjacency weights of edges such that $a_{ij}(t)>0\Leftrightarrow (j,i)\in \mathcal{E}(t)$ and $a_{ij}(t)=0$ otherwise. Denote $\mathcal{N}_{i}(t)=\{j\in\mathcal{V}|(j,i)\in\mathcal{E}(t)\}$ to represent the neighbor set at time $t$. The communication graph $\mathcal{G}(t)$ is said to be balanced if the sum of the interaction weights from and to an agent $i$ are equal, i.e., $ \sum\limits_{j=1}^{n}a_{ij}(t)=\sum\limits_{j=1}^{n}a_{ji}(t)$.
$(j,i)$ is called a $\delta-$edge if there always exist two positive constants $T$ and $\delta$ such that $\int_{t}^{t+T}a_{ij}(s)ds\geq\delta$ for any $t\geq0$. A $\delta-$graph, corresponding to $\mathcal{G}(t)$, is defined as $\mathcal{G}_{(\delta, T)}=(\mathcal{V},\mathcal{E}_{(\delta, T)})$, where $\mathcal{E}_{(\delta, T)}=\left\{(j,i)\in\mathcal{V}\times\mathcal{V}|\int_{t}^{t+T}a_{ij}(s)ds\geq\delta~for~any~t\geq0\right\}$. For a fixed topology $\mathcal{G}_{(\delta, T)}$, a path of length $r$ from node $i_1$ to node $i_{r+1}$ is a sequence of $r + 1$ distinct nodes $i_1 \cdots, i_{r+1}$ such that $(i_q, i_{q+1}) \in\mathcal{E}_{(\delta, T)}$ for $q=1,\cdots, r$. If there exists a path between any two nodes, then $\mathcal{G}_{(\delta, T)}$ is said to be strongly connected.

Here we make the following assumptions for the communication graph.
\begin{assumption}
The communication graph $\mathcal{G}(t)$ is balanced.
\end{assumption}
\begin{assumption}
The $\delta-$digraph $\mathcal{G}_{(\delta, T)}$ is strongly connected.
\end{assumption}

\subsection{Convex inequalities}
The objective of this note is to distributively search a feasible solution to the following inequalities:
\begin{equation}\label{eq1}
\textbf{g}(\textbf{x})\leq0
\end{equation}
where $\textbf{x}\in\mathbb{R}^{m}$ and $\textbf{g}(\cdot)=[g_{1}(\cdot), \cdots, g_{n}(\cdot)]^T$, each $g_{i}(\cdot):\mathbb{R}^{m}\rightarrow\mathbb{R}$ is a convex function which is only available to agent $i$. The following assumption is adopted throughout the paper.
\begin{assumption}
The feasible solution set of inequalities (\ref{eq1}) is non-empty.
\end{assumption}
Under Assumption 3, it is possible to search a point in $\textbf{X}=\left\{\textbf{x}|\textbf{g}(\textbf{x})\leq0\right\}$ over a network of agents. Now we introduce a plus function $g_{i}^{+}(\textbf{x})=\max [g_{i}(\textbf{x}), 0]$, $i=1,\cdots,n$. Note that if there exists a vector $\textbf{x}^{*}\in
\mathbb{R}^m$ such that $g_{i}^{+}(\textbf{x}^{*})=0$ for each $i\in\{1,\cdots,n\}$, then $\textbf{x}^{*}$ is a feasible solution to convex inequalities (\ref{eq1}).
Since functions $\max(\cdot)$ and $g_{i}(\cdot)$ are convex, function $g_{i}^{+}(\textbf{x})$ is also convex. Therefore, the subgradient of function $g_{i}^{+}(\textbf{x})$, denoted by $\nabla g_{i}^{+}(\textbf{x})$, always exists, and the following holds,
\begin{equation}\label{eq2}
 g_{i}^{+}(\textbf{y})-g_{i}^{+}(\textbf{x})\geq\langle\nabla g_{i}^{+}(\textbf{x}),\textbf{y}-\textbf{x}\rangle
\end{equation}
for any $\textbf{y}\in \mathbb{R}^{m}$.

Similar to \cite{A. Nedi07, A. Nedi11, Lin10}, we give the following assumption on the boundedness of $\nabla g_{i}^{+}(\textbf{x})$.
\begin{assumption}
  $\|\nabla g_{i}^{+}(\textbf{x})\|\leq K$ for some $K>0$, $i=1,\cdots,n$.
\end{assumption}

\subsection{Multi-agent systems for searching feasible solutions}
Now consider a continuous-time multi-agent system consisting of $n$ agents, labeled by set $\mathcal{V}=\{1,\cdots, n\}$. Each agent's dynamics is described as
\begin{equation}\label{eq3}
\dot{\textbf{x}_{i}}(t)=\textbf{u}_i(t),~~~~~i\in \mathcal{V}
\end{equation}
where $\textbf{x}_i(t), \textbf{u}_i(t)\in\mathbb{R}^{m}$ respectively represent the state and input of agent $i$. For convex inequalities (\ref{eq1}), the following subgradient-based algorithm is considered.
\begin{equation}\label{eq4}
\textbf{u}_i(t)=\sum\limits_{j\in\mathcal{N}_{i}(t)}a_{ij}(t)(\textbf{x}_j(t)-\textbf{x}_i(t))-b(t)\nabla g_{i}^{+}(\textbf{x}_i(t)),~~~~~i\in \mathcal{V}
\end{equation}
where $b(t)>0$ is a non-increasing function such that $\int_{0}^\infty b(t)dt\rightarrow\infty$ and $\int_{0}^\infty b^2(t)dt<\infty$.

From (\ref{eq3}) and (\ref{eq4}), the control input of agent $i$ is based on the subgradient information of the local plus function $g_i^+$ and the information received from its neighbors. Therefore, algorithm (\ref{eq4}) is distributed. Note that the positiveness of $b(t)$ and the boundedness of $\int_{0}^\infty b^2(t)dt$ imply that $\lim\limits_{t\rightarrow\infty}b(t)=0$. If we set $b(t)=0$ for any $t>0$, algorithm (\ref{eq4}) reduces to a standard ``consensus" or ``agreement" algorithm for continuous-time MASs in \cite{Martin20, S. Martin23}. {The conditions for $b(t)$ are actually constraints on its decaying rate, which guarantees convergence of the algorithm. This idea is inspired by the subgradient method \cite{Boyd14}. In particular, a suitable choice of $b(t)$ is $b(t)=\frac{a_0}{t+b_0}$ for any $t\geq0$, where $a_0$ and $b_0$ are two positive constants.}

In \cite{A. Nedi07, Lin10}, distributed subgradient-based algorithms were designed for discrete-time multi-agent systems to optimize a sum of convex objective functions. In this note, the agents are considered to have continuous-time dynamics. We aim to obtain conditions that not only guarantee consensus among all agents, but also ensure that the common state is a solution to the inequalities. The definition of consensus is stated as follows.
\begin{definition}
MAS (\ref{eq3}) is said to reach consensus asymptotically if $\lim_{t\rightarrow\infty}||x_i-x^*||=0$ for any $i\in\mathcal{V}$. $x^*$ is called the consensus state.
\end{definition}
%%%%%%%%%%%%%%%%%%%%%%%%%%%%%%%%%%%%%%%%%%%%%%%%%%%%%%%%%%%%%%%%%%%%%%%%%%%%%%%%%%%%%%%%%%%%%%%%%%
%%%%%%%%%%%%%%%%%%%%%%%%%%%%%%%%%%%%%%%%%%%%%%%%%%%%%%%%%%%%%%%%%%%%%%%%%%%%%%%%%%%%%%%%%%%%%%%%%%

\section{Main result}\label{se2}

Let us start this section by stating the main result, which indicates that, MAS (\ref{eq3}) with (\ref{eq4}) reaches consensus asymptotically and the convex feasibility problem is solvable.
\begin{theorem}
Under Assumptions 1-4, MAS (\ref{eq3}) with (\ref{eq4}) reaches consensus asymptotically and the consensus state is a feasible solution to convex inequalities (\ref{eq1}).
\end{theorem}
Now we define a time sequence $\left\{t_p\right\}_{p\in \mathbb{N}}$ with $t_0=0$. For any fixed $t_p$, $t_{p+1}$ is determined by another finite sequence $\left\{t_p^0, t_p^1, \cdots, t_p^{\lfloor n/2\rfloor}\right\}$, where $t_p^0=t_p$, $t_p^{\lfloor n/2\rfloor}=t_{p+1}$. Note that for any $h\in\left\{0, \cdots, \lfloor n/2\rfloor-1\right\}$, $t_p^{h+1}$ is the smallest time such that $t>t_p^{h}$ and
\begin{equation}\label{eq5}
\min_{S\subsetneqq\mathcal{V}, S\neq\mathbf{\varnothing}}\left\{\sum\limits_{i\in S}\sum\limits_{j\in S\setminus\mathcal{V}}\int_{t_p^{h}}^{t}a_{ij}(\tau)d\tau\right\}=1.
\end{equation}
Consider the following consensus model,
\begin{equation}\label{eq6}
 \dot{x_{i}}(t)=\sum\limits_{j\in\mathcal{N}_{i}(t)}a_{ij}(t)(x_j(t)-x_i(t)), ~~~~~ i\in \mathcal{V}
\end{equation}
where $x_{i}\in \mathbb{R}$. Let $\textbf{y}(t)=\left[x_{1}^T(t), \cdots, x_{n}^T(t)\right]^T$, system (\ref{eq6}) can be rewritten as
$$
 \dot{\textbf{y}}(t)=-\textbf{L}(t)\textbf{y}(t)
$$
where $\textbf{L}(t)=[\ell_{ij}(t)]\in \mathbb{R}^{n\times n}$ is the Laplacian matrix defined as $\ell_{ii}(t)=\sum\limits_{j=1}^{n}a_{ij}(t)$ and $\ell_{ij}(t)=-a_{ij}(t)$,$i\neq j$, see \cite{Godsil19} for detail. By the properties of linear systems \cite{Brockett22}, we have
 \begin{equation}\label{eq7}
\textbf{y}(t)=\mathbf{\Phi}(t,s)\textbf{y}(s)~~~~~ i\in \mathcal{V}
\end{equation}
where $\mathbf{\Phi}(t,s)$ is the state-transition matrix from state $\textbf{y}(s)$ to state $\textbf{y}(t)$ with $t\geq s\geq0$.
 Before giving the proof of our main result, some useful lemmas are needed. For consensus model (\ref{eq6}), the following lemma was proved by Martin and Girard \cite{S. Martin23}.
\begin{lemma}(Proposition 4 in \cite{S. Martin23})
Under Assumptions 3 and 4, let $P: \mathbb{R}^{+}\rightarrow\mathbb{N}$ be the function defined  by $P(t) = p$ if $t\in\left[t_{p}, t_{p+1}\right)$ for all $t>0$. If agent $i$ updates its state $x_{i}(t)$ with (\ref{eq6}), then
 \begin{equation}\label{eq8}\begin{split}
&\max_{1\leq i\leq n}x_{i}(t)-\min_{1\leq i\leq n}x_{i}(t)\leq \lambda^{P(t)-P(s)}
\left(\max_{1\leq i\leq n}x_{i}(s)-\min_{1\leq i\leq n}x_{i}(s)\right)\
\end{split}\end{equation}
for any $t\geq s\geq0$ and $x_{i}(s), i\in \mathcal{V}$, where $\lambda=1-\frac{1}{\left(8n^2\right)^{\lfloor n/2\rfloor}}$.
\end{lemma}
In fact, from the definition of sequence $\left\{t_p\right\}_{p\in \mathbb{N}}$ and the strong connectivity of $\delta-$graph, one could estimate the lower bound of the difference $P(t)-P(s)$ by $\frac{t-s}{(\lfloor 1/\delta\rfloor+1)\lfloor n/2\rfloor T}-1$ for any $t\geq s\geq0$. See the following lemma.
\begin{lemma}
Under Assumptions 3-4, for any $t\geq s\geq0$, if $x_{i}(s), i\in \mathcal{V}$ is updated by (\ref{eq6}), then it holds
\begin{equation}\label{eq9}\begin{split}
&\max_{1\leq i\leq n}x_{i}(t)-\min_{1\leq i\leq n}x_{i}(t)\leq H\gamma^{t-s}\left(\max_{1\leq i\leq n}x_{i}(s)-\min_{1\leq i\leq n}x_{i}(s)\right)
\end{split}\end{equation}
where $H=\lambda^{-1}$, $\gamma=\lambda^{\frac{1}{\left(\lfloor 1/\delta\rfloor+1\right)\left\lfloor n/2\right\rfloor T}}$ and $\lambda=1-\frac{1}{\left(8n^2\right)^{\lfloor n/2\rfloor}}$.
\end{lemma}
 \begin{proof} For any $i\in S\subsetneqq\mathcal{V}$, if $\delta-$graph is strongly connected, there must exist an agent $j\in S\setminus\mathcal{V}$ such that
 $(j, i)\in \mathcal{E}_{\delta}$. From the definition of $\delta-$edge, we have
 \[\begin{split}
&\min_{S\subsetneqq\mathcal{V}, S\neq\mathbf{\varnothing}}\left\{\sum\limits_{i\in S}\sum\limits_{j\in S\setminus\mathcal{V}}\int_{t_p^{h}}^{t_p^{h}+(\lfloor 1/\delta\rfloor+1)T}a_{ij}(s)ds\right\}\\
&\geq \min_{(j, i)\in \mathcal{E}_{\delta}}\int_{t_p^{h}}^{t_p^{h}+(\lfloor 1/\delta\rfloor+1)T}a_{ij}(s)ds\\
&\geq(\lfloor 1/\delta\rfloor+1)\delta\\
&\geq1.
\end{split}\]
 Since $t_p^{h+1}$ is the smallest time such that $t>t_p^{h}$ and equation (\ref{eq5}) holds, one has $t_p^{h+1}-t_p^{h}\leq (\lfloor 1/\delta\rfloor+1)T$ for any $h\in\left\{0, \cdots, \lfloor n/2\rfloor-1\right\}$. Therefore, $t_{p+1}-t_p\leq(\lfloor 1/\delta\rfloor+1)\lfloor n/2\rfloor T$, which implies $P(t)-P(s)\geq \frac{t-s}{(\lfloor 1/\delta\rfloor+1)\lfloor n/2\rfloor T}-1$. Recall that $0<\left(1-\frac{1}{\left(8n^2\right)^{\lfloor n/2\rfloor}}\right)<1$ in Lemma 1, it follows the validity of (\ref{eq9}). \end{proof}
\begin{lemma}
 Under Assumptions 3 and 4, for any $t\geq s\geq0$, the state-transition matrix in (\ref{eq7}) satisfies the following inequality
\begin{equation}\label{eq10}
\left|\left[\mathbf{\Phi}(t,s)\right]_{ij}-\frac{1}{n}\right|\leq H\gamma^{t-s},~~~~~i,j\in\{1,\cdots, n\}
\end{equation}
where $H=\lambda^{-1}$, $\gamma=\lambda^{\frac{1}{\left(\lfloor 1/\delta\rfloor+1\right)\left\lfloor n/2\right\rfloor T}}$ and $\lambda=1-\frac{1}{\left(8n^2\right)^{\lfloor n/2\rfloor}}$.
\end{lemma}
\begin{proof} Let $\textbf{e}_j$ be a standard unit base vector with the $j^{th}$ entry being one and others being zero. For any $j\in\{1,\cdots, n\}$, substituting $\textbf{y}(t)=\textbf{e}_j$ into equation (\ref{eq7}) yields $\textbf{y}(t)=\left[\mathbf{\Phi}(t,s)\right]_{\cdot j}$. From (\ref{eq9}) in Lemma 2, we have
$$
\max_{1\leq i\leq n}\left[\mathbf{\Phi}(t,s)\right]_{ij}-\min_{1\leq i\leq n}\left[\mathbf{\Phi}(t,s)\right]_{ij}\leq H\gamma^{t-s} ~~~~~j\in\{1,\cdots, n\}
$$
where $t\geq s\geq0$. Since the graph is balanced, $\sum_{i=1}^ny_i(t)$ is invariant, hence the average consensus is reached exponentially. This implies that $\left[\mathbf{\Phi}(\infty,s)\right]_{ij}=\frac{1}{n}$. By the fact that $\min_{1\leq i\leq n}x_i(t)$ is non-decreasing with (\ref{eq6}), it follows that
\begin{equation}\label{eq11}\begin{split}
\max_{1\leq i\leq n}\left[\mathbf{\Phi}(t,s)\right]_{ij}-\frac{1}{n} &=\max_{1\leq i\leq n}\left[\mathbf{\Phi}(t,s)\right]_{ij}-\min_{1\leq i\leq n}\left[\mathbf{\Phi}(\infty,s)\right]_{ij}\\
&\leq\max_{1\leq i\leq n}\left[\mathbf{\Phi}(t,s)\right]_{ij}-\min_{1\leq i\leq n}\left[\mathbf{\Phi}(t,s)\right]_{ij}\\
&\leq H\gamma^{t-s}.
\end{split}\end{equation}
Similarly, due to the fact that $\max_{1\leq i\leq n}x_i(t)$ is non-increasing, we can conclude
\begin{equation}\label{eq12}
\min_{1\leq i\leq n}\left[\mathbf{\Phi}(t,s)\right]_{ij}-\frac{1}{n}\geq -H\gamma^{t-s}.
\end{equation}
Inequalities (\ref{eq11}) and (\ref{eq12}) lead to the validity of (\ref{eq10}). \end{proof}
\begin{lemma}
Let $b(t)$ be a continuous function, if $\lim\limits_{t\rightarrow\infty}b(t)=b$ and $0<\gamma<1$, then $\lim\limits_{t\rightarrow\infty} \int_{0}^t \gamma^{t-s}$ $b(s)ds=-\frac{b}{ln\gamma}$.
\end{lemma}
\begin{proof} Since $\lim\limits_{t\rightarrow\infty}b(t)=b$, for arbitrary $\varepsilon>0$, there exists $T>0$ such that $|b(t)-b|\leq\varepsilon$ when $t\geq T$. Due to the fact that $b(t)$ is continuous, both the maximum value and the minimum value of $b(t)$ exist in closed interval $t\in[0,T]$. We denote them by $\overline{b}=\max_{0\leq t\leq T}b(t)$ and $\underline{b}=\min_{0\leq t\leq T}b(t)$, respectively. For $t\geq T$, we have
\[
\begin{split}
\int_{0}^t \gamma^{t-s}b(s)ds&=\int_{0}^T \gamma ^ {t-s}b(s)ds+\int_{T}^t \gamma^{t-s}b(s)ds\\
&\geq \underline{b}\int_{0}^T\gamma^{t-s}ds+(b-\varepsilon)\int_{T}^t\gamma^{t-s}ds\\
&=\frac{\underline{b}}{-ln\gamma}(\gamma^{t-T}-\gamma^{t})+\frac{b-\varepsilon}{-ln\gamma}(1-\gamma^{t-T}).
\end{split}
\]
Therefore, $\lim\inf_{t\rightarrow\infty} \int_{0}^t \gamma^{t-s}b(s)ds\geq\frac{b-\varepsilon}{-ln\gamma}$. Because $\varepsilon$ is arbitrary and fixed, we have $\lim\inf_{t\rightarrow\infty}$ $\int_{0}^t \gamma^{t-s}b(s)ds\geq\frac{b}{-ln\gamma}$. Similarly, we have
\[
\begin{split}
&\int_{0}^t \gamma^{t-s}b(s)ds\leq \frac{\overline{b}}{-ln\gamma}(\gamma^{t-T}-\gamma^{t})+\frac{b+\varepsilon}{-ln\gamma}(1-\gamma^{t-T}).
\end{split}
\]
Thus, it holds that $\lim\sup_{t\rightarrow\infty} \int_{0}^t \gamma^{t-s}b(s)ds\leq \frac{b+\varepsilon}{-ln\gamma}$. Due to the arbitrariness of $\varepsilon$, we have $\lim\sup_{t\rightarrow \infty} \int_{0}^t \gamma^{t-s}b(st)ds \leq\frac{b}{-ln\gamma}$. This and the fact $\lim\inf_{t\rightarrow\infty} \int_{0}^t \gamma^{t-s}b(s)ds\geq\frac{b}{-ln\gamma}$ imply that $\lim\limits_{t\rightarrow\infty} \int_{0}^t \gamma^{t-s}b(s)ds=-\frac{b}{ln\gamma}$. \end{proof}

Now, we can present the proof of Theorem 1.

\emph{Proof of  Theorem 1}. The proof is consisted of two parts. In part 1, we will prove that consensus can be achieved asymptotically by MAS (\ref{eq3}) with (\ref{eq4}). In part 2, we will be committed to showing that the state of each agent converges to the solution set of convex inequalities (\ref{eq1}). Now let us begin with the first part.

\textbf{Part 1.} We define a vector $\widetilde{\textbf{x}}_\mu(t)\in \mathbb{R}^{n}$ which stacks up the $\mu^{th}$ entry of $\textbf{x}_i(t)$, $i\in \textbf{V}$, in other words, the $j^{th}$ entry of vector $\widetilde{\textbf{x}}_\mu(t)\in \mathbb{R}^{n}$ is the $\mu^{th}$ entry of $\textbf{x}_j(t)$. Similarly, we also define vector $ \textbf{f}_{\mu}(t) \in \mathbb{R}^{n}$ to be the vector stacking up the $\mu^{th}$ entry of $\nabla g_{i}^{+}(\textbf{x}_i(t))$, $i\in \textbf{V}$. From (\ref{eq3}) and (\ref{eq4}), we have
$$
 \dot{\widetilde{\textbf{x}}}_\mu(t)=-\textbf{L}(t)\widetilde{\textbf{x}}_\mu(t)-b(t)\textbf{f}_{\mu}(t)
$$
where $\mu=1, 2, \cdots, m$. The term $-b(t)\textbf{f}_{\mu}(t)$ can be viewed as a control input of the linear system. By the basic properties of linear systems \cite{Brockett22}, we have
\begin{equation}\label{eq13}
\widetilde{\textbf{x}}_\mu(t)=\mathbf{\Phi}(t,0)\widetilde{\textbf{x}}_\mu(0)-\int_{0}^{t}b(\tau)\mathbf{\Phi}(t,\tau)\textbf{f}_{\mu}(\tau)d\tau
\end{equation}
where $\mathbf{\Phi}(t, 0)$ is the state-transition matrix. By Peano-Baker formula (see \cite{Brockett22} for detail), it can be concluded that $\mathbf{\Phi}(t,s)$ is a double-stochastic matrix under Assumption 1. Then, equation (\ref{eq13}) further implies that
\begin{equation}\label{eq14}
\textbf{1}^T\widetilde{\textbf{x}}_\mu(t)=\textbf{1}^T\widetilde{\textbf{x}}_\mu(0)-\int_{0}^{t}b(\tau)\textbf{1}^T\textbf{f}_{\mu}(\tau) d\tau.
\end{equation}
Note that $\textbf{1}^T\widetilde{\textbf{x}}_\mu(t)$ is a scalar. On the basis of (\ref{eq13}) and (\ref{eq14}), we have
\[\begin{split}
&\left|[\widetilde{\textbf{x}}_\mu(t)]_{i}-\frac{1}{n}\textbf{1}^T\widetilde{\textbf{x}}_\mu(t)\right|\\
&=\Big|\left([\mathbf{\Phi}(t,0)]_{i\cdot}-\frac{1}{n}\textbf{1}^T\right)\widetilde{\textbf{x}}_\mu(0)
-\int_{0}^{t}b(\tau)\left([\mathbf{\Phi}(t,\tau)]_{i\cdot}-\frac{1}{n}\textbf{1}^T\right)\textbf{f}_{\mu}(\tau)d\tau\Big|\\
&\leq\left|\left([\mathbf{\Phi}(t,0)]_{i\cdot}-\frac{1}{n}\textbf{1}^T\right)\widetilde{\textbf{x}}_\mu(0)
\right|+\int_{0}^{t}b(\tau)\left|\left([\mathbf{\Phi}(t,\tau)]_{i\cdot}-\frac{1}{n}\textbf{1}^T\right)\textbf{f}_{\mu}(\tau)\right|d\tau\\
&\leq\max_{1\leq j\leq n}\left|[\mathbf{\Phi}(t,0)]_{ij}-\frac{1}{n}\right|\left\|\widetilde{\textbf{x}}_\mu(0)\right\|_1
+K\int_{0}^{t}b(\tau)\max_{1\leq j\leq n}\left|[\mathbf{\Phi}(t,\tau)]_{ij}-\frac{1}{n}\right|d\tau\\
\end{split}\]
for every $i=1,\cdots, n$. By Lemma 3, it follows
\begin{equation}\label{eq15}\begin{split}
\left|[\widetilde{\textbf{x}}_\mu(t)]_{i}-\frac{1}{n}\textbf{1}^T\widetilde{\textbf{x}}_\mu(t)\right|\leq H\gamma^t\left\|\widetilde{\textbf{x}}_\mu(0)\right\|_1+KH\int_{0}^{t}b(\tau)\gamma^{t-\tau}d\tau.\\
\end{split}\end{equation}
Because $0<\gamma<1$ and $\lim\limits_{t\rightarrow\infty}b(t)=0$, it follows from Lemma 4 that $\lim\limits_{t\rightarrow\infty}\left|[\widetilde{\textbf{x}}_\mu(t)]_{i}-\frac{1}{n}\textbf{1}^T\widetilde{\textbf{x}}_\mu(t)\right|=0$ for any $i=1,\cdots, n$. This means that the limits of the $\mu^{th}$ entries of all $\textbf{x}_i(t)$ are equal. Note that each component of $\textbf{x}_i(t)$ is decoupled in (\ref{eq4}). Therefore, consensus is reached for any $\mu\in\{1,\cdots,m\}$, implying that MAS (\ref{eq3}) with (\ref{eq4}) reaches consensus asymptotically, i.e., $\lim\limits_{t\rightarrow\infty}\|\textbf{x}_i(t)-\textbf{x}^*\|=0$ for any $i\in\mathcal{V}$.

\textbf{Part 2.} For ease of description, we denote the average value of all $\textbf{x}_i(t)$ by $\bar{\textbf{x}}(t)=\frac{1}{n}\sum\limits_{i=1}^{n}\textbf{x}_{i}(t)$. From (\ref{eq15}), we can further conclude that
\begin{equation}\label{eq16}\begin{split}
\left\|\textbf{x}_{i}(t)-\bar{\textbf{x}}(t)\right\|\leq &H\sqrt{m}\gamma^t\max_{1\leq\mu\leq m}\left\|\widetilde{\textbf{x}}_\mu(0)\right\|_1+KH\sqrt{m}\int_{0}^{t}b(\tau)\gamma^{t-\tau}d\tau,~~~~~i\in \mathcal{V}
\end{split}\end{equation}
where $\widetilde{\textbf{x}}_\mu(t)$ is defined as Part 1. Because $b(t)$ is non-increasing and positive, it holds that $\int_{0}^{\infty}b(s)\gamma^sds\leq b(0)\int_{0}^{\infty}\gamma^sds=\frac{b(0)}{-ln\gamma}$. Furthermore, we have
\begin{equation}\label{eq17}\begin{split}
\int_{0}^{\infty}\int_{0}^{s}b(s)b(\tau)\gamma^{t-\tau}d\tau ds&=\int_{0}^{\infty}\int_{0}^{s}\gamma^{\theta}b(s)b(s-\theta)d\theta ds\\
&=\int_{0}^{\infty}\gamma^{\theta}\int_{\theta}^{\infty}b(s)b(s-\theta)ds d\theta\\
&\leq \int_{0}^{\infty}\gamma^{\theta}\int_{\theta}^{\infty}b(s-\theta)b(s-\theta)ds d\theta\\
&=\frac{1}{-ln\gamma}\int_{0}^{\infty}b^2(s)ds\\
&<\infty
\end{split}\end{equation}
where the first equality holds by letting $t-\tau=\theta$, the second one results by changing the order of the
integrals, and the first inequality comes from the fact that $b(t)$ is non-increasing. Therefore, by inequalities (\ref{eq16}) and (\ref{eq17}), we have
\begin{equation}\label{eq18}\begin{split}
\int_{0}^{\infty}b(t)\left\|\textbf{x}_{i}(t)-\bar{\textbf{x}}(t)\right\|dt&\leq \frac{b(0)H\sqrt{m}}{-ln\gamma}\max_{1\leq\mu\leq m}\left\|\widetilde{\textbf{x}}_\mu(0)\right\|_1+\frac{KH\sqrt{m}}{-ln\gamma}\int_{0}^{\infty}b^2(t)dt\\
&<\infty
\end{split}\end{equation}
for any $i\in \mathcal{V}$. MAS (\ref{eq3}) with (\ref{eq4}) can be rewritten as
$$
\dot{\textbf{x}}(t)=-(\textbf{L}(t)\otimes I)\textbf{x}-b(t)\mathbf{\nabla}(t)
$$
where $\textbf{x}(t)=\left[\textbf{x}_1^T(t), \cdots, \textbf{x}_n^T(t)\right]^T$ and $\mathbf{\nabla}(t)=\left[[\nabla g_1^+(\textbf{x}_1(t))]^T, \cdots, [\nabla g_n^+(\textbf{x}_n(t))]^T\right]^T$. From Assumption 1, $\textbf{1}^T\textbf{L}=0$. Let $\bar{\textbf{x}}(t)=\frac{1}{n}(\textbf{1}^T\otimes I)\textbf{x}(t)$, we have
\begin{equation}\label{eq19}
\dot{\bar{\textbf{x}}}(t)=-\frac{1}{n}b(t)(\textbf{1}^T\otimes I)\mathbf{\nabla}(t)=-\frac{1}{n}b(t)\sum\limits_{i=1}^{n}\nabla g_i^+(\textbf{x}_i(t)).
\end{equation}
Now consider a function $d(\cdot):\mathbb{R}^{m}\rightarrow\mathbb{R}$ given by $d(\bar{\textbf{x}}(t))=\frac{1}{2}\|\bar{\textbf{x}}(t)-\textbf{x}_0\|^2$, where $\textbf{x}_0\in\textbf{X}$.
Along with equation (\ref{eq19}), taking the derivative of function $d$ with respect to $t$ yields
\begin{equation}\label{eq20}\begin{split}
\dot{d}(\bar{\textbf{x}}(t))&=-\frac{1}{n}b(t)\sum\limits_{i=1}^{n} \left\langle\nabla g_i^+(\textbf{x}_i(t)), \bar{\textbf{x}}(t)-\textbf{x}_0\right\rangle\\
&=-\frac{1}{n}b(t)\sum\limits_{i=1}^{n} \left\langle\nabla g_i^+(\textbf{x}_i(t)), \bar{\textbf{x}}(t)-\textbf{x}_i(t)\right\rangle+\frac{1}{n}b(t)\sum\limits_{i=1}^{n} \left\langle\nabla g_i^+(\textbf{x}_i(t)), \textbf{x}_0-\textbf{x}_i(t)\right\rangle.
\end{split}\end{equation}
Due to the property of bounded subgradients, it holds that $\|g_i^+(\textbf{x})-g_i^+(\textbf{y})\|\leq K\|\textbf{x}-\textbf{y}\|$ for arbitrary vectors $\textbf{x}, \textbf{y}\in \mathbb{R}^{m}$. By the fact that function $g_i^+(\cdot)$ is convex and $\textbf{\emph{g}}^+(\textbf{x}_0)=\textbf{0}$, it follows from inequality (\ref{eq2}) that
 \begin{equation}\label{eq21}\begin{split}
\left\langle\nabla g_i^+(\textbf{x}_i(t)), \textbf{x}_0-\textbf{x}_i(t)\right\rangle&\leq-g_i^+(\textbf{x}_i(t))\\
&=-g_i^+(\bar{\textbf{x}}(t))+g_i^+(\bar{\textbf{x}}(t))- g_i^+(\textbf{x}_i(t))\\
&\leq-g_i^+(\bar{\textbf{x}}(t))+K\|\bar{\textbf{x}}(t)-\textbf{x}_i(t)\|.
\end{split}\end{equation}
 Then, combining (\ref{eq20}) and (\ref{eq21}), we have
\begin{equation}\label{eq22}\begin{split}
\dot{d}(\bar{\textbf{x}}(t))\leq \frac{2K}{n}b(t)\sum\limits_{i=1}^{n} \left\|\textbf{x}_i(t)- \bar{\textbf{x}}(t)\right\|
-\frac{1}{n}b(t)\sum\limits_{i=1}^{n} g_i^+(\bar{\textbf{x}}(t)).
\end{split}\end{equation}
Integrating both sides of inequality (\ref{eq22}) over $[0, t]$ for any $t\geq0$ yields
\begin{equation}\label{eq23}\begin{split}
d(\bar{\textbf{x}}(t))-d(\bar{\textbf{x}}(0))&\leq \frac{2K}{n}\sum\limits_{i=1}^{n} \int_{0}^{t}b(\tau)\left\|\textbf{x}_i(\tau)- \bar{\textbf{x}}(\tau)\right\|d\tau-\frac{1}{n}\sum\limits_{i=1}^{n}\int_{0}^{t}b(\tau)g_i^+(\bar{\textbf{x}}(\tau))d\tau.
\end{split}\end{equation}
Now we denote function $ h(t)=\sum\limits_{i=1}^{n}\int_{0}^{t}b(\tau)\left\|\textbf{x}_i(\tau)- \bar{\textbf{x}}(\tau)\right\|d\tau$. It is obvious that $ h(t)$ is non-decreasing with respect to $t$. Inequality (\ref{eq18}) shows that $h(t)$ is upper bounded. This implies that $h(t)$ converges, i.e., there exists a $0\leq h^*<\infty$ such that $\lim\limits_{t\rightarrow\infty}h(t)=h^*$. Furthermore, for any $t_{1}>t_{2}>0$, it holds that $\int_{0}^{t_{1}}b(\tau)g_i^+(\bar{\textbf{x}}(\tau))d\tau\geq\int_{0}^{t_{2}}b(\tau)g_i^+(\bar{\textbf{x}}(\tau))d\tau$ due to the fact that $g_i^+(\bar{\textbf{x}}(\tau))\geq0$ and $b(t)\geq0$. Therefore, one has
$$
d(\bar{\textbf{x}}(t_{1}))-d(\bar{\textbf{x}}(t_{2}))\leq \frac{2K}{n}(h(t_{1})-h(t_{2})).
$$
This implies $\lim\limits_{t\rightarrow\infty}\sup d(\bar{\textbf{x}}(t))-\lim\limits_{t\rightarrow\infty}\inf d(\bar{\textbf{x}}(t))\leq\frac{2K}{n} (\lim\limits_{t\rightarrow\infty}\sup h(t)-\lim\limits_{t\rightarrow\infty}\inf h(t))=0$. As a result, $\lim\limits_{t\rightarrow\infty}d(\bar{\textbf{x}}(t))$ exists. On the other hand, by (\ref{eq23}), we have
\begin{equation}\label{eq24}\begin{split}
\frac{1}{n}\sum\limits_{i=1}^{n}\int_{0}^{\infty}b(t)g_i^+(\bar{\textbf{x}}(t))dt\leq \frac{2K}{n}h^*+d(\bar{\textbf{x}}(0))-d(\bar{\textbf{x}}(\infty))<\infty.
\end{split}\end{equation}
Note that $g_i^+(\bar{\textbf{x}}(t))$ is non-negative for any $i\in\mathcal{V}$ and $t>0$. Then, by (\ref{eq24}) and the fact $\int_{0}^\infty b(t)dt\rightarrow\infty$, we have $\lim\limits_{t\rightarrow\infty}\inf g_i^+(\bar{\textbf{x}}(t))=0$ for any $i\in \mathcal{V}$. Thus, there exists a subsequence $\{\bar{\textbf{x}}(t_{k})\}$ of $\{\bar{\textbf{x}}(t)\}$ that converges to a point in the solution set of convex inequality (\ref{eq1}). Without loss of generality, assume that $\textbf{x}^*$ is this point. We have $\lim\limits_{k\rightarrow\infty}\bar{\textbf{x}}(t_{k})=\lim\limits_{t\rightarrow\infty}\inf\bar{\textbf{x}}(t)=\textbf{x}^*\in\textbf{X}$. Moreover, let $x_0=x^*$ in $d(t)$,  the fact that $d(t)$ converges implies $\lim\limits_{t\rightarrow\infty}\bar{\textbf{x}}(t)=\textbf{x}^*\in\textbf{X}$. Recall that the result in Part 1 implies $\lim\limits_{t\rightarrow\infty}\|\textbf{x}_i(t)-\bar{\textbf{x}}_j(t)\|=0$ for any $i\in\mathcal{V}$. Hence, we have $\lim\limits_{t\rightarrow\infty}\textbf{x}_i(t)=\textbf{x}^*\in\textbf{X}$ for any $i\in\mathcal{V}$.
\QEDA
\begin{remark}
In Theorem 1, the case when the solution set is non-empty is discussed. In fact, throughout the proof, it is not difficult to draw a conclusion that if the convex inequalities' solution set $\textbf{X}$ is empty, MAS (\ref{eq3}) with (\ref{eq4}) will reach consensus asymptotically, and each agent's state converges to a common state $\textbf{x}^*$ such that $\sum\limits_{i=1}^{n} g_i^+(\textbf{x}^*)=\min \sum\limits_{i=1}^{n} g_i^+(\textbf{x})$. Note that even if the solution set of inequalities (\ref{eq1}) is empty, the first part of the proof remains to be valid. Hence, consensus is asymptotically reached. Now we show that the consensus state minimizes the function $\sum\limits_{i=1}^{n} g_i^+$. For the sake of simplicity, we denote $f^*=\min \sum\limits_{i=1}^{n} g_i^+(\textbf{x})$. Then, (\ref{eq21}) should be replaced by $\langle\nabla g_i^+(\textbf{x}_i(t)), $ $\textbf{x}_0-\textbf{x}_i(t)\rangle\leq-(g_i^+(\bar{\textbf{x}}(t))-g_i^+(\textbf{x}_0))+K\|\bar{\textbf{x}}(t)-\textbf{x}_i(t)\|$, where $\textbf{x}_0$ is another point such that $\sum\limits_{i=1}^{n} g_i^+(\textbf{x}_0)=f^*$. As a result, inequality (\ref{eq24}) is replaced by $\frac{1}{n}\int_{0}^{\infty}b(t)(\sum\limits_{i=1}^{n}g_i^+(\bar{\textbf{x}}(t))-f^*)dt<\infty$, it can be concluded that $\lim\limits_{t\rightarrow\infty}\textbf{x}_i(t)=\textbf{x}^*$, and $\sum\limits_{i=1}^{n} g_i^+(\textbf{x}^*)=f^*$.
\end{remark}

%%%%%%%%%%%%%%%%%%%%%%%%%%%%%%%%%%%%%%%%%%%%%%%%%%%%%%%%%%%%%%%%%%%%%%%%%%%%%%%%%%%%%%%%%%%%%%%%%%
%%%%%%%%%%%%%%%%%%%%%%%%%%%%%%%%%%%%%%%%%%%%%%%%%%%%%%%%%%%%%%%%%%%%%%%%%%%%%%%%%%%%%%%%%%%%%%%%%%

\section{Distributed optimization with convex inequality constraints}\label{se3}

Now we extend our method to solving a constrained optimization problem. Different from the problem of optimizing the sum of local objective functions subject to the intersection of constraint sets in \cite{Nedi08, Lin10}, our goal is to distributively minimize the objective function subject to convex inequalities, which is stated as follow.
\begin{equation}\label{eq25}\begin{split}
&\min~~~~~\sum\limits_{i=1}^{n}f_i(\textbf{x})\\
&\textrm{subject}~\textrm{to} ~~~ \textbf{g}(\textbf{x})\leq0
\end{split}\end{equation}
where $\textbf{x}\in\mathbb{R}^{m}$ and $\textbf{g}(\cdot)=[g_{1}(\cdot), \cdots, g_{n}(\cdot)]$, both $f_{i}(\cdot)$ and $g_{i}(\cdot):\mathbb{R}^{m}\rightarrow\mathbb{R}$ are convex functions. Agent $i$ can only have access to $f_{i}(\cdot)$ and $g_{i}(\cdot)$. The following assumptions are made in this section.
\begin{assumption}
The set $\{\textbf{x}\in\mathbb{R}^{m}|\textbf{g}(\textbf{x})\leq0\}$ is non-empty.
\end{assumption}
\begin{assumption}
$\|\nabla f_{i}(\textbf{x})\|\leq K_1$ and $\|\nabla g_{i}(\textbf{x})\|\leq K_2$ for some $K_1, K_2>0$, $i=1,\cdots,n$.
\end{assumption}

{Assumption 5 implies Slater's constraint qualification condition holds \cite{M. S. Bazaraa24}, then the solution set of problem (\ref{eq25}) is guaranteed to be non-empty. A Lagrange function of problem (\ref{eq25}) is defined as}
\[
F(\textbf{x},\textbf{z})=\sum\limits_{i=1}^{n}F_i(\textbf{x},z_i); F_i(\textbf{x},z_i)=f_i(\textbf{x})+z_ig_i(\textbf{x})
\]
{where $\textbf{z}=[z_1,\cdots, z_n]^T$ is the Lagrange multiplier such that $\textbf{z}\geq0$. It is obvious that $F_i(\textbf{x},z_i)$ is convex with $\textbf{x}$ and linear with $z_i$ for any $i=1,\cdots, n$. Thus, $F_i(\textbf{x},z_i)$ is a convex-concave function and so is F(\textbf{x},\textbf{z}). Based on Saddle-point Theorem \cite{M. S. Bazaraa24}, we know that $\textbf{x}^*$ is an optimal solution of (\ref{eq25}) if and only if there exists a positive vector $\textbf{z}^*\in \mathbb{R}^{n}$ such that $(\textbf{x}^*, \textbf{z}^*)$ is a saddle point of $F(\textbf{x}, \textbf{z})$, i.e., $F(\textbf{x}^*, \textbf{z})\leq F(\textbf{x}^*, \textbf{z}^*)\leq F(\textbf{x}, \textbf{z}^*)$ for any $\textbf{x}\in\mathbb{R}^{m}$ and $\textbf{z}\in\mathbb{R}^{n}$. For ease, we use $\textbf{X}^*\times \textbf{Z}^*$ to represent the saddle point set, where $\textbf{X}^*$ denotes the optimal solution set of (\ref{eq25}) and $\textbf{Z}^*$ denotes the corresponding optimal set of Lagrange multipliers.}

{Before extending (\ref{eq4}) for searching the optimal solution to (\ref{eq25}), we introduce the following compact and convex sets}
\[\Omega_i=\{z\in\mathbb{R}|0\leq z\leq\bar{z}_i\},~~~i=1, \cdots, n \]
{where each $\bar{z}_i$ is a finite positive real number, and can be sufficiently large. Denote the Cartesian product of $\Omega_i, i=1, \cdots, n$ by $\Omega$, i.e., $\Omega=\Omega_1\times\cdots\times \Omega_n$. Now we extend (\ref{eq4}) for problem (\ref{eq25}) as follows:}
\begin{equation}\label{eq26}\begin{split}
\textbf{u}_{i}(t)=&\sum\limits_{j\in\mathcal{N}_{i}(t)}a_{ij}(t)(\textbf{x}_j(t)-\textbf{x}_i(t))-b(t) (\nabla f_i(\textbf{x}_i(t))+z_i(t)\nabla g_i(\textbf{x}_i(t)))~~~i\in \mathcal{V}
\end{split}\end{equation}
{where $b(t)$ is defined as (4) and $z_i(t)\in\mathbb{R}$ is an auxiliary variable, whose dynamic is given as}
\[
\dot{z}_i(t)=P_{\mathcal{T}_{\Omega_i}(z_i(t))}[b(t)g_i(\textbf{x}_i(t))]
\]
{where $\mathcal{T}_{\Omega_i}(z_i(t))$ is the tangent cone of $\Omega_i$ at point $z_i(t)$, and the initial value is set to be $z_i(0)=z_{i0}\in\Omega_i$. By the definition of $\Omega_i$, it is not difficult to compute that $P_{\mathcal{T}_{\Omega_i}(z_i(t))}[b(t)g_i(x_i(t))]=0$ if $z_i(t)=0$, $b(t)g_i(x_i(t))<0$, or $z_i(t)=\overline{z}_i$, $b(t)g_i(\textbf{x}_i(t))>0$; $P_{\mathcal{T}_{\Omega_i}(z_i(t))}[b(t)g_i(\textbf{x}_i(t))]=b(t)g_i(\textbf{x}_i(t))$ otherwise.}
{\begin{theorem}
Under Assumptions 1, 2, 5 and 6, if $\Omega\cap \textbf{Z}^*$ is non-empty, then MAS (\ref{eq3}) with (\ref{eq26}) reaches consensus asymptotically and the consensus state is an optimal solution to (\ref{eq25}).
\end{theorem}}
{\begin{proof}{Let $\textbf{x}(t)=\left[\textbf{x}_1^T(t), \cdots, \textbf{x}_n^T(t)\right]^T$, $\mathbf{\nabla}_i(t)$
$=\nabla f_i(\textbf{x}_i(t))+z_i(t)\nabla g_i(\textbf{x}_i(t))$ for $i=1, \cdots, n$ and $\mathbf{\nabla}(t)=\left[[\mathbf{\nabla}_1(t)]^T, \cdots, [\mathbf{\nabla}_n(t)]^T\right]^T$, then MAS (\ref{eq3}) with (\ref{eq26}) can be rewritten as}
$$
\dot{\textbf{x}}(t)=-(\textbf{L}(t)\otimes I)\textbf{x}-b(t)\mathbf{\nabla}(t).
$$
{Note that $z_i(t)\in\Omega$ holds for any $t\geq0$, similar to inequality (\ref{eq16}), it can be concluded that}
\[\begin{split}
\left\|\textbf{x}_{i}(t)-\bar{\textbf{x}}(t)\right\|\leq & H\sqrt{m}\gamma^t\max_{1\leq\mu
\leq m}\left\|\widetilde{\textbf{x}}_\mu(0)\right\|_1+\bar{K}H\sqrt{m}\int_{0}^{t}b(\tau)\gamma^{t-\tau}d\tau~~~~~i\in \mathcal{V}.
\end{split}\]
{where $\bar{K}=K_1+K_2\max_{1\leq\mu\leq n}\{z_i\}$. By Lemma 4, we have $\lim\limits_{t\rightarrow\infty}\|\textbf{x}_i(t)-\bar{\textbf{x}}(t)\|=0$ for any $i\in\mathcal{V}$. Thus, MAS (\ref{eq3}) with (\ref{eq26}) reaches consensus asymptotically. Through a similar approach to those in (\ref{eq17}) and (\ref{eq18}), it follows that $\int_{0}^{\infty}b(t)\left\|\textbf{x}_{i}(t)-\bar{\textbf{x}}(t)\right\|dt<\infty$. Furthermore, consider the function}
 \[\begin{split}
 &d(\bar{\textbf{x}}(t), \textbf{z}(t))=d_1(\bar{\textbf{x}}(t))+d_2(\textbf{z}(t)), \\
 &d_1(\bar{\textbf{x}}(t))=\frac{1}{2}\|\bar{\textbf{x}}(t)-\textbf{x}_0\|^2,\\
 &d_2(\textbf{z}(t))= \frac{1}{2n}\|\textbf{z}(t)-\textbf{z}_0\|^2\end{split}\]
{ where $\textbf{x}_0\in\textbf{X}^*$ and $\textbf{z}_0\in\Omega\cap \textbf{Z}^*$, $\textbf{z}_0=[z_{10},\cdots, z_{n0}]^T$. It is obvious that $(\textbf{x}_0, \textbf{z}_0)\in\textbf{X}^*\times \textbf{Z}^*$. Since $F_i(\textbf{x},z_i)$ is convex with respect to $\textbf{x}$. Similar to (\ref{eq20})-(\ref{eq22}), we have}
\begin{equation}\label{eq27}\begin{split}
 \dot{d}_1(\bar{\textbf{x}}(t))\leq& \frac{2\bar{K}}{n}b(t)\sum\limits_{i=1}^{n} \left\|\textbf{x}_i(t)- \bar{\textbf{x}}(t)\right\|-\frac{1}{n}b(t)(F(\bar{\textbf{x}}(t), \textbf{z}(t))-F(\textbf{x}_{0}, \textbf{z}(t))).
\end{split}\end{equation}
{Moreover, $\dot{z}_i(t)=P_{\mathcal{T}_{\Omega_i}(z_i(t))}[b(t)g_i(\textbf{x}_i(t))]$ implies that there exists an element $c_i(z_i(t))\in\mathcal{C}_{\Omega_i}(z_i(t))$ such that $\dot{z}_i(t)=b(t)g_i(\textbf{x}_i(t))-c_i(z_i(t))$, where $\mathcal{C}_{\Omega_i}(z_i(t))=\{d|d(z'-z_i(t))\leq0, \forall z'\in\Omega_i\}$ is the
normal cone of $\Omega_i$ at element $z_i(t)\in\Omega_i$ (see \cite{Brogliato30} for detail). Thus, we have}
\[\begin{split}
 \dot{d}_2(\bar{\textbf{x}}(t))&\leq\frac{1}{n}\sum\limits_{i=1}^{n}(z_i(t)-z_{i0})(b(t)g_i(\textbf{x}_i(t))-c_i(z_i(t)))\\
 &\leq\frac{1}{n}b(t)\sum\limits_{i=1}^{n}(z_i(t)-z_i0)g_i(\textbf{x}_i(t))\\
 &=\frac{1}{n}b(t)\sum\limits_{i=1}^{n}(F_i(\textbf{x}_i(t),z_i(t))-F_i(\textbf{x}_i(t),z_{i0}))\\
 &=\frac{1}{n}b(t)\sum\limits_{i=1}^{n}((F_i(\textbf{x}_i(t),z_i(t))-F_i(\bar{\textbf{x}}(t),z_i(t)))+(F_i(\bar{\textbf{x}}(t),z_i(t))-F_i(\bar{\textbf{x}}(t),z_{i0}))).\\
 \end{split}\]
{ Note that for any element $z_i\in\Omega_i$, it holds $|F_i(\textbf{x}_i(t),z_i)-F_i(\bar{\textbf{x}}(t),z_i)|\leq\bar{K}\|\textbf{x}_i(t)-\bar{\textbf{x}}(t)\|$. Then}
\begin{equation}\label{eq28}\begin{split}
 \dot{d}_2(\bar{\textbf{x}}(t))\leq &\frac{2\bar{K}}{n}b(t)\sum\limits_{i=1}^{n}\|\textbf{x}_i(t)-\bar{\textbf{x}}(t)\|+\frac{1}{n}b(t) (F(\bar{\textbf{x}}(t),\textbf{z}(t))-F(\bar{\textbf{x}}(t),\textbf{z}_{0})).\\
 \end{split}\end{equation}
{Together with (\ref{eq27}), we have}
\[\begin{split}
\dot{d}(\bar{\textbf{x}}(t), \textbf{z}(t))\leq&\frac{4\bar{K}}{n}b(t)\sum\limits_{i=1}^{n} \left\|\textbf{x}_i(t)- \bar{\textbf{x}}(t)\right\|-\frac{1}{n}b(t)(F(\bar{\textbf{x}}(t),\textbf{z}_{0})-F(\textbf{x}_{0}, \textbf{z}(t))).
 \end{split}\]
{Similar to the proof of Theorem 1, it can be concluded that}
\[\begin{split}
&\int_{0}^{\infty}b(t)(F(\bar{\textbf{x}}(t),\textbf{z}_{0})-F(\textbf{x}_{0}, \textbf{z}_0))dt+ \int_{0}^{\infty}b(t)(F(\textbf{x}_{0}, \textbf{z}_0)-F(\textbf{x}_{0},\textbf{z}(t)))dt<\infty.
\end{split}\]
{Since $F(\bar{\textbf{x}}(t),\textbf{z}_{0})-F(\textbf{x}_{0}, \textbf{z}_0)$ and $F(\textbf{x}_{0}, \textbf{z}_0)-F(\textbf{x}_{0},\textbf{z}(t))$ are both non-negative, it holds that $\int_{0}^{\infty}b(t)(F(\bar{\textbf{x}}(t),\textbf{z}_{0})-F(\textbf{x}_{0}, \textbf{z}_0))dt<\infty$ and $\int_{0}^{\infty}b(t)(F(\textbf{x}_{0}, \textbf{z}_0)-F(\textbf{x}_{0},\textbf{z}(t)))dt<\infty$. Together with the fact that $\int_{0}^\infty b(t)dt\rightarrow\infty$, we can conclude that $\lim\limits_{t\rightarrow\infty}(F(\bar{\textbf{x}}(t),\textbf{z}_{0})-F(\textbf{x}_{0}, \textbf{z}_0))=0$ and $\lim\limits_{t\rightarrow\infty}(F(\textbf{x}_{0}, \textbf{z}_0)-F(\textbf{x}_{0},\textbf{z}(t)))=0$. This implies that there exists a vector $\textbf{x}^*\in\textbf{X}^*$ such that $\lim\limits_{t\rightarrow\infty}\textbf{x}_i(t)=\textbf{x}^*$ for any $i\in\mathcal{V}$. }\end{proof}
%%%%%%%%%%%%%%%%%%%%%%%%%%%%%%%%%%%%%%%%%%%%%%%%%%%%%%%%%%%%%%%%%%%%%%%%%%%%%%%%%%%%%%%%%%%%%%%%%%
%%%%%%%%%%%%%%%%%%%%%%%%%%%%%%%%%%%%%%%%%%%%%%%%%%%%%%%%%%%%%%%%%%%%%%%%%%%%%%%%%%%%%%%%%%%%%%%%%%

\section{Simulations}\label{se4}

In this section, we give numerical examples to illustrate the obtained results.

\begin{example}Consider nine agents with the index set $\{1,\cdots,9\}$. The agents communicate with each other via a time-varying directed graph, which periodically switches between two subgraphs depicted in Fig.\ref{fig1} with period $T=0.3$, and the weight of each edge is set to be 1. Algorithm (\ref{eq4}) is used for searching a feasible solution to inequalities $g_i(\textbf{x})=\sum\limits_{j=1}^{3} c_{ij}x_{j}+d_{j}$, $i=1,\cdots,9$, where $c_{11}=2, c_{12}=3, c_{13}=4; c_{21}=2, c_{22}=-3, c_{23}=-4; c_{31}=-2, c_{32}=1, c_{33}=0.5; c_{41}=2, c_{42}=-1, c_{43}=6; c_{51}=1, c_{52}=0, c_{53}=2; c_{61}=1, c_{62}=-2, c_{63}=0.3; c_{71}=0.5, c_{72}=2, c_{73}=1; c_{81}=-1, c_{82}=-1, c_{83}=0.5; c_{91}=-2, c_{92}=3, c_{93}=3$ and $d_1=-0.1;d_2=-3;\textcolor{blue}{d_3=-1};d_4=2;d_5=1;d_6=-1;d_7=-2;d_8=0.1;d_9=1$. Since $g_i(\textbf{x})$ is linear, the inequalities with three variables are convex. We denote $R_{i}(t)=\|\textbf{x}_{i}(t)-\frac{1}{9}\sum\limits_{j=1}^{9}\textbf{x}_{j}(t)\|,~~ i\in\mathcal{V}$ and $Q(t)=\sum\limits_{j=1}^{9}g_i^+(\textbf{x}_i(t))$. Let each agent's initial state equal to the same vector $[1, -0.5, 1]^T$ and $b(t)=\frac{0.9}{t+5}$, the trajectories of $R_i(t)$ and $Q(t)$ are shown in Fig. \ref{fig2} and Fig. \ref{fig3}, respectively. Fig. \ref{fig2} indicates that $\textbf{x}_i(t)$ converges to a common point $\textbf{x}^*$ for any $i\in\mathcal{V}$ as $t\rightarrow\infty$. It is computed that $\textbf{x}^*=[0.13, -0.15, -0.57]^T$.  Fig. \ref{fig3} shows that $\textbf{g}(\textbf{x}^*)\leq 0$. These observations are consistent with the results established in Theorem 1.
\end{example}

\begin{example}{Consider five agents with the index set $\{1,\cdots,5\}$. The communication graph is shown in Fig. \ref{fig4} and the weight of each edge equals 1. Algorithm (\ref{eq26}) is used for solving optimization problem (\ref{eq25}) with $\textbf{x}\in \mathbb{R}$, where the local cost functions are given as follows:}
\[
f_i(\textbf{x})=\left\{ {\begin{array}{*{20}c}\begin{split}
   &{0.5\omega_i\textbf{x}^2,~~~-100\leq\textbf{x}\leq 100}  \\
   &{-100\omega_i\textbf{x}+1.5\times10^4\omega_i,~~~\textbf{x}< -100 }  \\
   &{100\omega_i\textbf{x}-0.5\times10^4\omega_i,~~~\textbf{x}>100 }  \\
\end{split}\end{array}} \right.\begin{array}{*{20}c}
   {} \\
\end{array}i=1, \cdots, 5
\]
{where $\omega_1=0.5, \omega_2=0.3, \omega_3=0.4, \omega_4=0.6, \omega_5=0.2$. Note that $|\nabla f_i(\textbf{x})|\leq100\omega_i$. Given inequality constraints $g_1(\textbf{x})=2\textbf{x}-8, g_2(\textbf{x})=-\textbf{x}+2,g_3(\textbf{x})=\textbf{x}-4.5,g_4(\textbf{x})=3\textbf{x}-15$ and $g_5(\textbf{x})=-\textbf{x}+1$, it can be easily verified that the optimal solution is $\textbf{x}^*=2$. Let $\bar{z}_i=50$ for any $i = 1, \cdots ,5$, $b(t)=\frac{2.6}{2t+0.25}$, and the initial states of agents be $\textbf{x}_1(0)=3, \textbf{x}_2(0)=-2, \textbf{x}_3(0)=-1, \textbf{x}_4(0)=1$ and $\textbf{x}_5(0)=3$,  Fig. \ref{fig5} shows that the states of all the agents converge to the same optimal solution $\textbf{x}^*=2$. This is consistent with the result established in Theorem 2.}
\end{example}
\begin{figure}[htbp]
\centering
\includegraphics[width=0.6\textwidth]{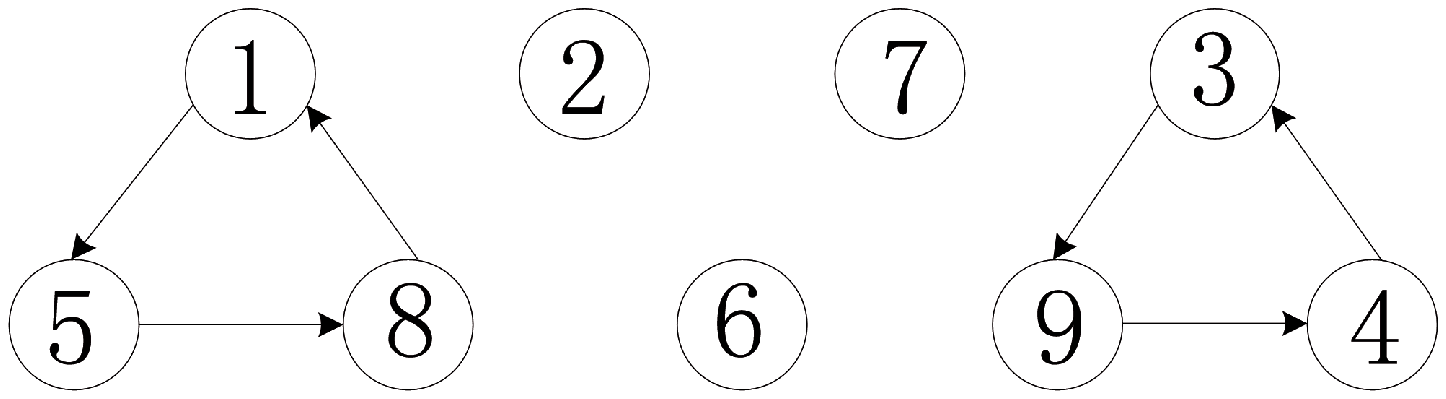}
~~~~~\includegraphics[width=0.6\textwidth]{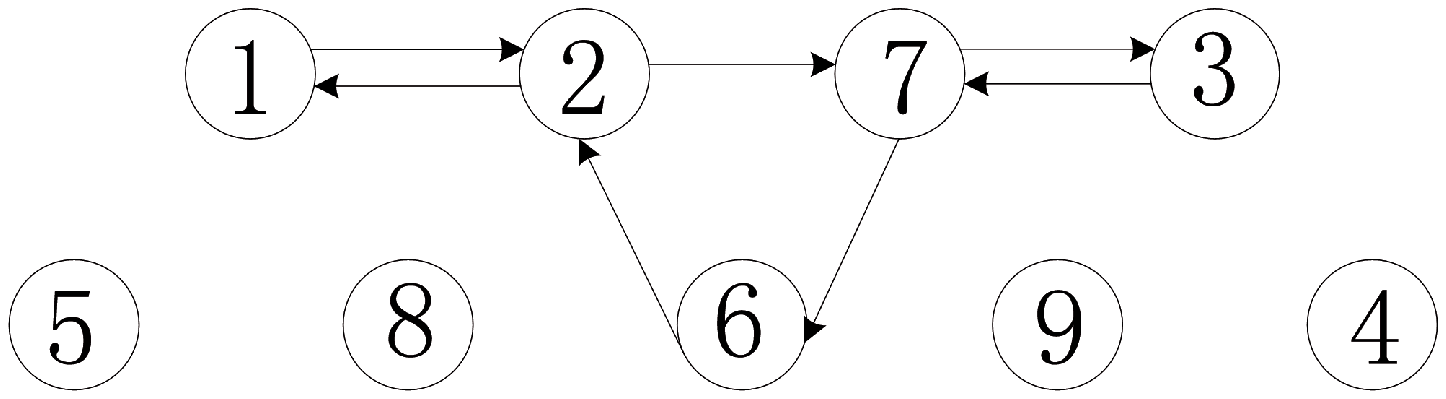}
\caption{The top subfigure represents the communication topology at time interval $t\in [0.3k,0.3k+0.15)$, where $k\in\mathbb{N}$. The bottom one represents the communication topology at time interval $t\in [0.3k+0.15,$ $0.3k+0.3)$.}\label{fig1}
\end{figure}

\begin{figure}
\begin{minipage}[t]{0.5\linewidth}
\centering
\includegraphics[width=0.8\textwidth]{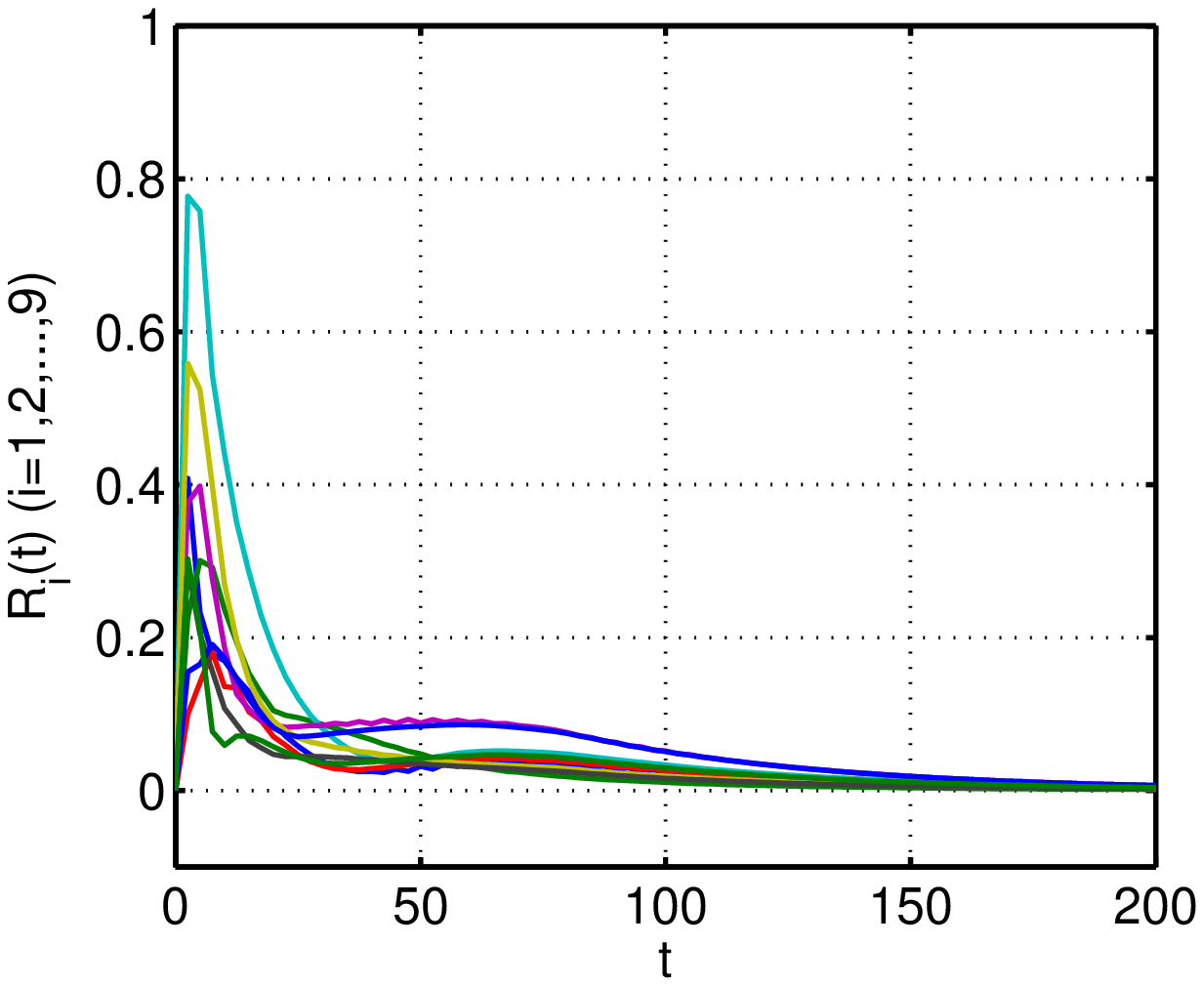}
\caption{The trajectories of $R_{i}(t)$ when algorithm (\ref{eq4}) is implemented, $i\in\mathcal{V}$.}\label{fig2}
\end{minipage}
\begin{minipage}[t]{0.5\linewidth}
\centering
\includegraphics[width=0.8\textwidth]{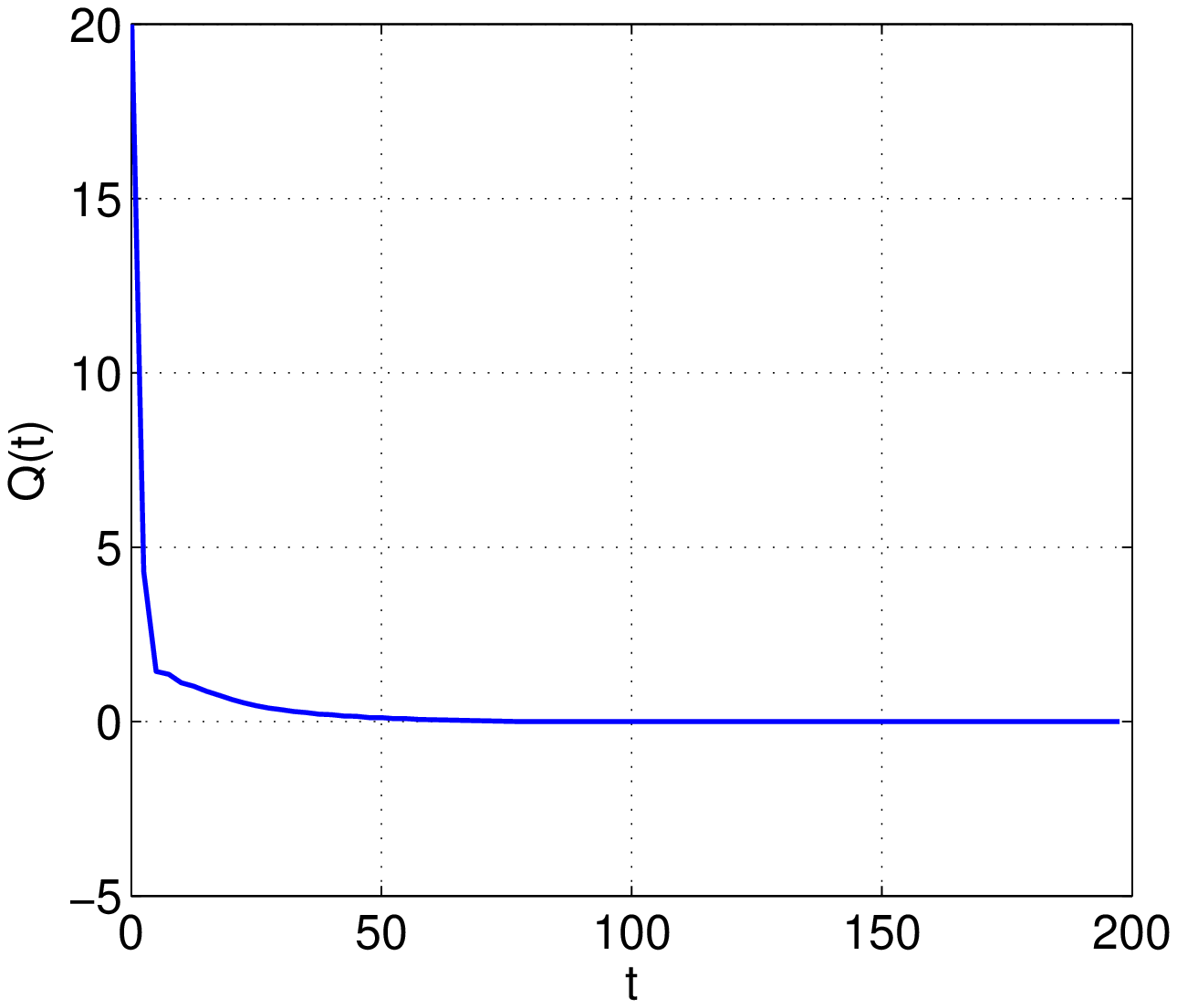}
\caption{The trajectories of $Q(t)$ when algorithm (\ref{eq4}) is implemented.}\label{fig3}
\end{minipage}
\end{figure}
\begin{figure}[htbp]
\centering
\includegraphics[width=0.42\textwidth]{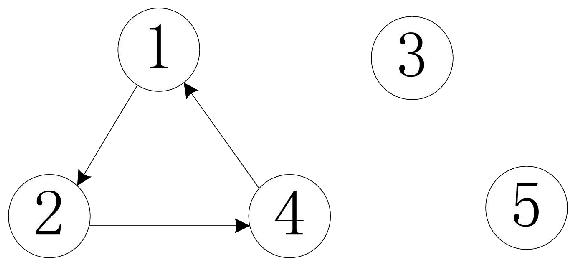}
~~~~~\includegraphics[width=0.42\textwidth]{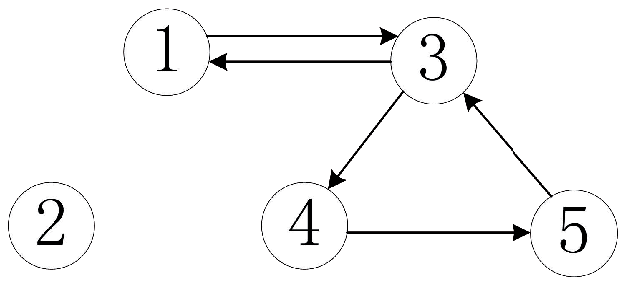}
\caption{The left subfigure represents the communication topology at time interval $t\in [k,k+0.5)$, where $k\in\mathbb{N}$. The right one represents the communication topology at time interval $t\in [k+0.5,$ $k+1)$.}\label{fig4}
\end{figure}

\begin{figure}
\centering
\includegraphics[width=0.5\textwidth]{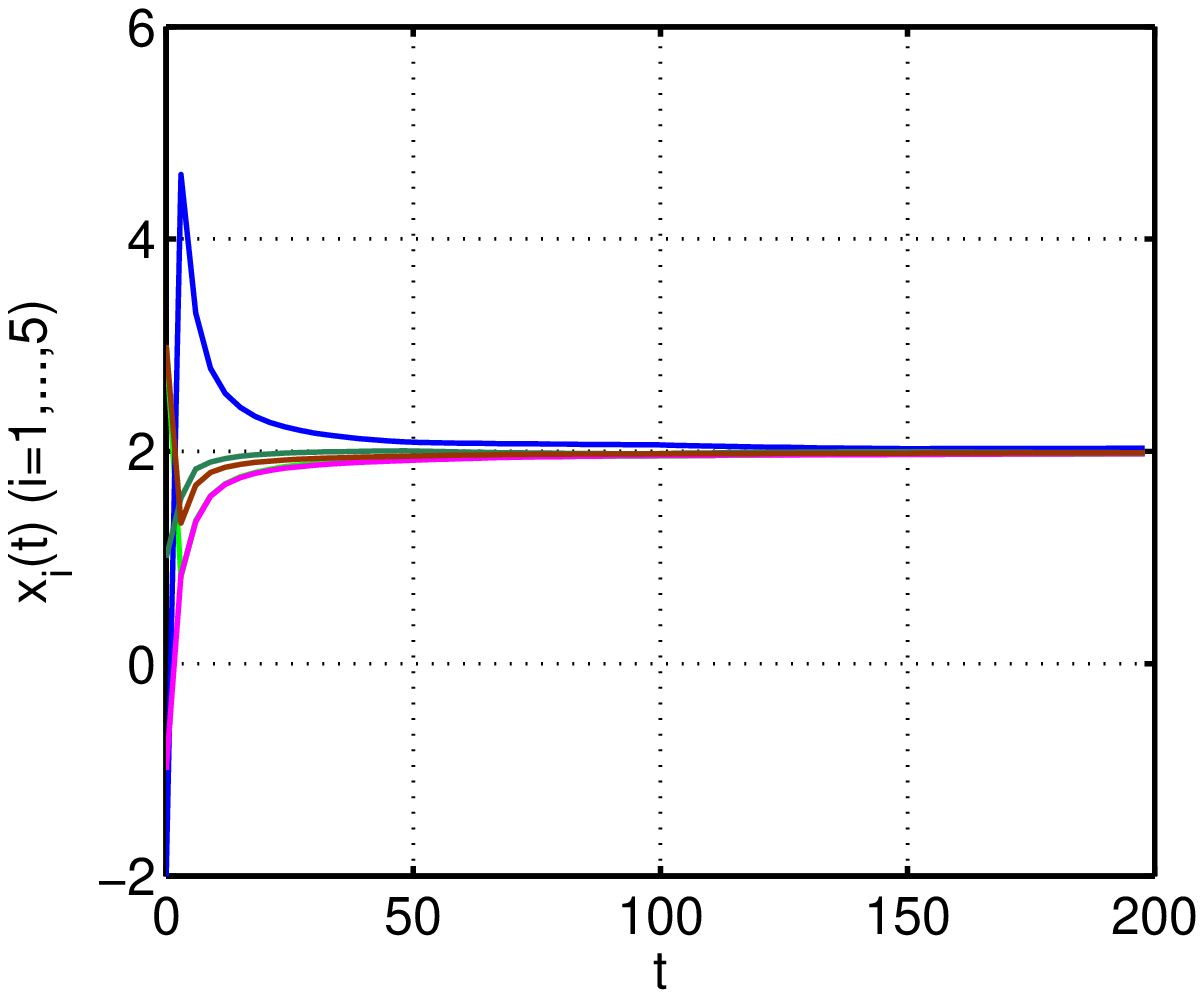}
\caption{The trajectories of $\textbf{x}_i(t)$ when algorithm (\ref{eq26}) is implemented.} \label{fig5}
\end{figure}

\section{Conclusion}\label{se5}
In this note, we have presented a continuous-time distributed computation model to search a feasible solution to convex inequalities. In this model, each agent adjusts its state value based on local information received from its immediate neighbors and its own inequality information using a subgradient method. It is shown that if the $\delta-$graph, induced by a time-varying directed graph, is strongly connected, the multi-agent system will reach a common state asymptotically and the consensus state is a feasible solution to convex inequalities. The method has been effectively extended to solving the distributed optimization problem of minimizing the sum of local objective functions subject to convex inequalities. Simulation examples have been conducted to demonstrate the effectiveness of our results. Our future work will focus on some other interesting topics, such as the case with time delays, packet loss and communication bandwidth constraints, which will bring new challenges in searching feasible solutions to inequalities over a network of agents.

%%%%%%%%%%%%%%%%%%%%%%%%%%%%%%%%%%%%%%%%%%%%%%%%%%%%%%%%%%%%

% if have a single appendix:
%\appendix[Proof of the Zonklar Equations]
% or
%\appendix  % for no appendix heading
% do not use \section anymore after \appendix, only \section*
% is possibly needed

% use appendices with more than one appendix
% then use \section to start each appendix
% you must declare a \section before using any
% \subsection or using \label (\appendices by itself
% starts a section numbered zero.)
%
%\appendices
%\section{Proof of the First Zonklar Equation}
%Appendix one text goes here.

% you can choose not to have a title for an appendix
% if you want by leaving the argument blank
%\section{}
%Appendix two text goes here.

% use section* for acknowledgement

% Can use something like this to put references on a page
% by themselves when using endfloat and the captionsoff option.
\ifCLASSOPTIONcaptionsoff
  \newpage
\fi

% trigger a \newpage just before the given reference
% number - used to balance the columns on the last page
% adjust value as needed - may need to be readjusted if
% the document is modified later
%\IEEEtriggeratref{8}
% The "triggered" command can be changed if desired:
%\IEEEtriggercmd{\enlargethispage{-5in}}

% references section

% can use a bibliography generated by BibTeX as a .bbl file
% BibTeX documentation can be easily obtained at:
% http://www.ctan.org/tex-archive/biblio/bibtex/contrib/doc/
% The IEEEtran BibTeX style support page is at:
% http://www.michaelshell.org/tex/ieeetran/bibtex/
%\bibliographystyle{IEEEtran}
% argument is your BibTeX string definitions and bibliography database(s)
%\bibliography{IEEEabrv,../bib/paper}
%
% <OR> manually copy in the resultant .bbl file
% set second argument of \begin to the number of references
% (used to reserve space for the reference number labels box)

\end{document}